\documentclass{llncs}

\usepackage{amsmath}
\usepackage{amssymb}
\usepackage{xspace}
\usepackage{bm}
\usepackage{graphicx}
\usepackage{multirow}
\usepackage{pdflscape}
\usepackage{array}
\usepackage{relsize}
\usepackage[noadjust]{cite}

\usepackage{xcolor}
\usepackage{algorithmicx}
\usepackage{algorithm}
\usepackage[noend]{algpseudocode}
\algrenewcommand\alglinenumber[1]{\color{gray}\sf\footnotesize#1}

\usepackage{thmtools}

\spnewtheorem{observation}{Observation}{\itshape}{\rmfamily}

\newcommand{\cL}{\mathcal{L}}

\newcommand{\Lup}{L^{\mathsf{rem}}}
\newcommand{\Ldown}{L^{\mathsf{add}}}

\newcommand{\addcontainer}[1]{\mathsf{add}(#1)}
\newcommand{\remcontainer}[1]{\mathsf{rem}(#1)}
\newcommand{\nummoves}[1]{\mathsf{\#moves}(#1)}

\newcommand{\dfn}{:=}

\newcommand{\ie}{i.e.}

\newcommand{\lb}{$\mathrm{LB}$\xspace}
\newcommand\Cpp{C\nolinebreak[4]\hspace{-.05em}\raisebox{.4ex}{\relsize{-3}{\textbf{++}}}\xspace}

\newcommand{\configurationstacking}{\textsc{Configuration Stacking}\xspace}
\newcommand{\prioritystacking}{\textsc{Priority Stacking}\xspace}
\newcommand{\exclusionscheduling}{\textsc{Mutual Exclusion Scheduling}\xspace}

\newcolumntype{L}[1]{>{\raggedright\let\newline\\\arraybackslash\hspace{0pt}}m{#1}}
\newcolumntype{C}[1]{>{\centering\let\newline\\\arraybackslash\hspace{0pt}}m{#1}}
\newcolumntype{R}[1]{>{\raggedleft\let\newline\\\arraybackslash\hspace{0pt}}m{#1}}

\title{A branch and price procedure for the container premarshalling problem\thanks{Part of this paper will appear in the proceedings of ESA 2014.}}
\author{Martijn van Brink \inst{1} and Ruben van der Zwaan \inst{2}} 
\institute{Maastricht University, email: \email{m.vanbrink@maastrichtuniversity.nl} \and Eindhoven University of Technology, email: \email{g.r.j.v.d.zwaan@tue.nl}}

\date{\today}

\begin{document}
\maketitle

\begin{abstract}
During the loading phase of a vessel, only the containers that are on top of their stack are directly accessible. If the container that needs to be loaded next is not the top container, extra moves have to be performed, resulting in an increased loading time. One way to resolve this issue is via a procedure called premarshalling. The goal of premarshalling is to reshuffle the containers into a desired lay-out prior to the arrival of the vessel, in the minimum number of moves possible. This paper presents an exact algorithm based on branch and bound, that is evaluated on a large set of instances. The complexity of the premarshalling problem is also considered, and this paper shows that the problem at hand is NP-hard, even in the natural case of stacks with fixed height.
\end{abstract}

\section{Introduction}
Enormous volumes of goods are shipped yearly all over the world in standardized containers. These containers typically require multiple modes of transportation to reach their destination. At container terminals, containers are transshipped between ships, trucks, and trains. This transshipment generally does not occur immediately upon delivery of a container, therefore containers are temporarily stored in an area called the container yard. The container yard consists of a set of blocks, which in turn consist of a set of bays. Each bay contains a number of rows, called stacks, with a certain height.

One main indicator of the efficiency of a container terminal is the berthing time of a vessel, which consists primarily of the time needed to load and unload containers. During the unloading phase, information on pick-up time and destination of the containers is often inaccurate or even unknown. This makes it difficult to obtain an unloading sequence that permits an efficient loading sequence. Hence, during the loading phase it can occur that the container that needs to be retrieved next, is not on top of the stack. In this case, the containers on top of this container need to be rehandled, \ie, relocated within the container yard, before the desired container can be retrieved. These rehandle operations greatly increase the time needed to remove the container from the yard.

One way to resolve this issue is to reshuffle the containers prior to the arrival of the vessel. This operation is called \emph{remarshalling}, and the goal is to find a sequence of rehandles, also called moves, of minimum length that reorganizes the stacks such that no container that is needed early is below a container that is needed late. This results is no rehandles during the loading phase, thus reducing the berthing time. The only valid move is to pick up the top container of one stack and put it on top of another stack, see Figure~\ref{fig:examplestacking} for an illustration.

\begin{figure}
  \centering
  \includegraphics[width=0.6\textwidth]{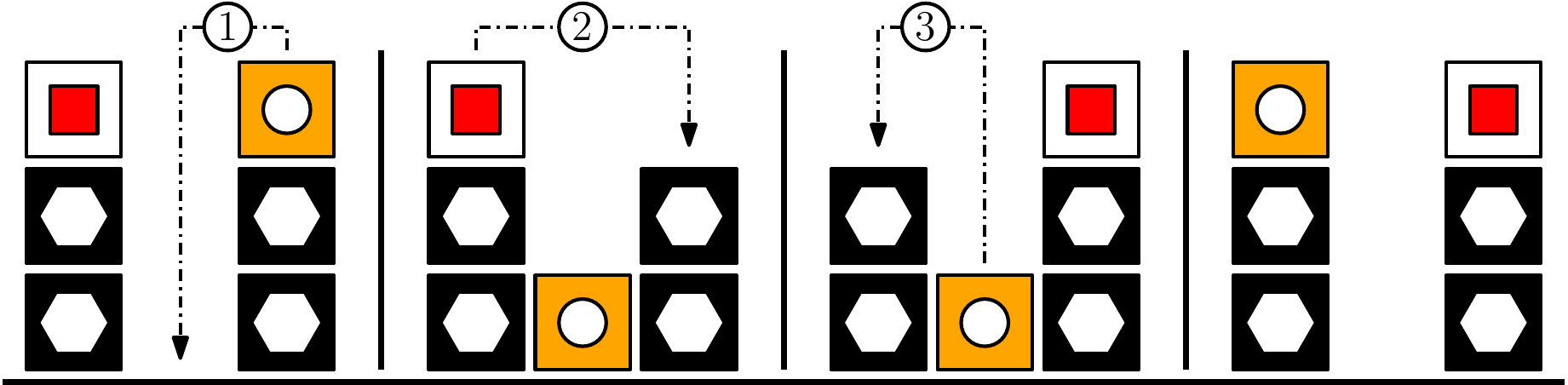}\\
  \caption{Small example of three moves transforming the stacks. }\label{fig:examplestacking}
\end{figure}

Two types of remarshalling operations can be identified, called \emph{intra-block remarshalling} (hereafter called remarshalling) and \emph{intra-bay premarshalling} (hereafter called premarshalling). The containers are reshuffled between bays for the former, and within a bay for the latter. The premarshalling variant is primarily applicable to container yards that use rail mounted gantry cranes. For safety reasons, it is not allowed to move the crane to a different bay while carrying a container. If a container needs to be moved to another bay, it is first placed on a truck, after which the truck moves to the target bay. There, a second crane picks up the container from the truck. This operation is extremely time consuming, and is thus avoided whenever possible \cite{LC2009}. Another main difference between remarshalling and premarshalling lies in the number of cranes that are used. For premarshalling typically only a single crane is used, while for remarshalling several cranes are often used simultaneously \cite{CSV2011}.

In this paper we focus on the premarshalling problem, and we follow the same assumptions as Bortfeldt and Forster \cite{BF2012}: (a) a single crane is used for rehandling containers, and (b) the time needed to move a container from one stack to another does not depend on the distance between the two stacks. This last assumption follows from the fact that the time needed to position the crane over a stack is negligible compared to the time needed to pick up or drop off a container. As a consequence, we are only interested in the number of rehandling operations.

We assume that for each container a priority level is given, and the goal is to transform the initial lay-out into a desired target lay-out in the minimum number of rehandles. The main problem studied is  called \prioritystacking: we accept all lay-outs in which no container with a lower priority is placed on top of a container with a higher priority. In a variant, called \configurationstacking, we restrict the target lay-out to a single pre-specified lay-out. The main motivation to also look at \configurationstacking is that giving a concrete target lay-out might give guidance to the algorithm and yield faster computation time. 

\paragraph{Related literature}
The operations at container terminals are well studied in the literature. Steenken, Vo{\ss}, and Stahlbock \cite{SVS2004} and Stahlbock and Vo{\ss} \cite{SV2008} describe the most important processes and operations at container terminals and give an overview of methods to optimize these operations. Vis and De Koster \cite{VK2003} give a classification of the different decision problems that occur at container terminals, and give an overview of relevant literature.

While there is a vast amount of work on the logistics of container terminals, the number of publications on the premarshalling problem is limited. We are only aware of one paper that provides an exact algorithm. Lee and Hsu \cite{LH2007} develop a mixed integer linear program based on a multicommodity network flow formulation that solves both \prioritystacking and \configurationstacking to optimality. However, this formulation can only be reasonably applied to very small instances, and the running time heavily depends on the choice of the number of time points. For larger instances the authors provide a heuristic that iteratively applies the exact approach on small parts of the instance. Integer multicommodity network flow is a generalization of edge-disjoint paths which cannot be approximated better than $\Omega(\sqrt{n})$, which immediately implies that the integrality gap of this formulation is at least that \cite{GKRSY2003}.

The remaining literature on the premarshalling problem is about the design of fast heuristics for \prioritystacking. Lee and Chao \cite{LC2009} describe a heuristic that minimizes the weighted sum of the mis-overlay index, which can be seen as a measure for the number of rehandles during the loading phase, and the number of rehandles during the premarshalling phase. Caserta and Vo\ss \cite{CV2009b} develop a heuristic based on the corridor method, where the basic idea is to use an exact method for limited portions of the entire solution space. Bortfeldt and Forster \cite{BF2012} describe a refined heuristic tree search procedure that looks at move sequences rather than individual moves. This heuristic is reported to be faster than the heuristic by Caserta and Vo\ss.

Caserta, Schwarze, and Vo{\ss} \cite{CSV2011} give an overview of recent developments on three so-called post-stacking problems. Besides the remarshalling and premarshalling problem, the authors also consider the (intra-bay) blocks relocation problem. In addition to the premarshalling problem, containers need to be removed from the bay in a certain order that minimizes the number of rehandles. It was proven that this problem is NP-hard, but with arbitrarily high stacks \cite{CSV2012}. This proof also works as a proof that \prioritystacking is NP-hard with arbitrarily high stacks. To the best of the authors knowledge there are no results about the natural case when stack heights are bounded by a constant. Typical stack heights are between $2$ and $8$ containers, while currently used equipment can handle a stack height of at most $10$ containers \cite{VK2003, SV2008}.

\paragraph{Our contributions}
We develop a fast exact algorithm based on column generation for the premarshalling problem and evaluate it extensively. To the best of our knowledge, we are the first to extensively experiment with an exact algorithm. Lee and Hsu \cite{LH2007} also design an exact algorithm, but only evaluate it on two instances. Our algorithm is evaluated on $960$ instances, with roughly $70\%$ of the instances solved within one second. We also see that our method exhibits a low integrality gap. Finally, we consider the complexity of \prioritystacking and \configurationstacking. Current NP-hardness proofs require a stack height that depends on the number of containers. We state an NP-completeness proof for both problems where the stack height is constant, which resembles the real-life situation.

\paragraph{Organization}
In Section \ref{sec:preliminaries} we introduce notation and formally describe the premarshalling problem. In Section \ref{sec:ilp} we describe an ILP formulation and an oracle for finding variables in a column generation approach. This is then used in Section \ref{sec:bnp} to design a branch and price algorithm, whose experimental performance is analyzed in Section \ref{sec:expRes}. In Section \ref{sec:nphard} we consider the complexity of both premarshalling variants. Finally, some conclusions are drawn in Section \ref{sec:conclusion}.

\section{Preliminaries}\label{sec:preliminaries}
Let $[n] := \{1,...,n\}$. For two intervals $[a,b]$ and $[c,d]$ we say that they \emph{overlap} if $[a,b] \cap [c,d] \neq \emptyset$ and neither interval contains the other interval. A set of intervals is {\em non-overlapping} if the intervals pairwise do not overlap.

The premarshalling problem is defined as follows. Given are $m$ stacks of maximum height $h$ and $n$ containers, each container labeled with a priority $\ell$ from $[k]$. In line with the definitions used in the literature, a lower priority number indicates a higher priority level, i.e., containers with priority $1$ are needed first, and containers with priority $k$ last. The \emph{lay-out} of a stack $i$ with $j \leq h$ containers is denoted as an ordered set of priorities $X_i := \{ x_1,\ldots,x_j \}$, where the first element is the priority of the bottom container and the last element is the priority of the top container. Notice that containers with the same priority are indistinguishable, therefore we will abbreviate ``move container with priority $\ell$'' to ``move container $\ell$''.

The goal is to transform the initial lay-out to a target lay-out by performing the minimum number of moves, while adhering to the maximum stack height. A move is defined as picking up the top-most container of one stack and placing it on top of another stack. For \prioritystacking the set of target lay-outs consists of all lay-outs such that all stacks are sorted in non-increasing order when viewed from the bottom, \ie, for a stack $i$ with $X_{i} := \{x_{1}, \ldots, x_{j}\}$, we have that $x_{p+1} \leq x_{p}$ for $p = 1,\ldots,j-1$. For \configurationstacking there is only one target lay-out, which is specified beforehand.

\section{Formulation as an ILP}\label{sec:ilp}
In this section we describe the linear program model that we use for both \prioritystacking and \configurationstacking.

Let us first introduce some notation. Let the tuple $(\ell,t)$ denote a move of container $\ell$ at time $t$, and consider stack $s$. Let $\Ldown_{s}$ and $\Lup_{s}$ contain moves $(\ell,t)$ such that container $\ell$ is respectively added to, or removed from, stack $s$ at time $t$, and let $L_{s} \dfn ( \Ldown_{s} , \Lup_{s})$. The set $L_{s}$ is \emph{feasible} if (a) at every time $t$ there is at most one move; (b) for all moves $(\ell,t) \in \Ldown_{s}$ and $(\ell,t) \in \Lup_{s}$ we have that at time $t$ stack $s$ contains at least one free spot, or container $\ell$ is the top container of stack $s$, respectively; (c) the lay-out obtained by executing the moves is a target lay-out. Hence, a feasible $L_{s}$ can be viewed as a sequence of moves that transforms stack $s$ into a target lay-out, where $\Ldown_{s}$ and $\Lup_{s}$ consist of the moves where a container is added or removed, respectively. Furthermore, let $\nummoves{L_{s}}$ denote the number of containers added in $L_{s}$, \ie, $\nummoves{L_{s}} \dfn |\Ldown_{s}|$. Finally, Let $\cL_{s}$ denote the set of sequences that transform stack $s$ into a target lay-out, \ie, $\cL_{s} \dfn \{L_{s} : L_{s} \text{ is feasible}\}$.

We consider the following ILP, where $x_{s,L} \in \{0,1\}$ for $L \in \cL_{s}$ has value $1$ if and only if stack $s$ is transformed according to sequence $L$. Let $\addcontainer{L,\ell,t}$ and $\remcontainer{L,\ell,t}$ be equal to $1$ if and only if $(\ell,t) \in \Ldown$ and $(\ell,t) \in \Lup$, respectively.

\begin{alignat}{3} \tag{$\mathsf{Integer~Linear~Program}$}\label{ILP}
\text{min} \quad & \sum_{s \in [m]} \sum_{L \in \cL_s} \nummoves{L} \ x_{s,L}  & \\
\text{s.t.} \quad & \sum_{s \in [m]} \sum_{L \in \cL_s} \left( \addcontainer{L,\ell,t}-\remcontainer{L,\ell,t}\right) x_{s,L}  \geq 0 & \quad \forall \ell \in [k], t \in [T] \tag{C1} \label{addRemove} \\
& \sum_{s \in [m]} \sum_{L \in \cL_s}\sum_{\ell \in [k]} \addcontainer{L,\ell,t} x_{s,L} \leq 1 & \forall t \in [T] \tag{C2} \label{oneMove} \\
& \sum_{L \in \cL_s} x_{s,L} = 1 & \forall s \in [m] \tag{C3} \label{oneSelected} \\
& x_{s,L} \in \{0,1\} &  \quad \forall s \in [m], L \in \cL_s \notag
\end{alignat}

The variables themselves already ensure that only sequences of moves are chosen that {\em for every stack individually} is feasible. The remaining constraints ensure that the local solutions together form a valid global solution.
Constraint~\eqref{addRemove} ensures that at time $t$ the number of containers of priority $\ell$ that are added is at least the number of containers of priority $\ell$ that are removed.
Constraint~\eqref{oneMove} enforces that at any time $t$ at most one container can be added. Note that this implies that at most one container is moved per time point.
Constraint~\eqref{oneSelected} make sure that exactly one sequence is selected for each stack.
By relaxing the requirement that the variables are either $0$ or $1$ we obtain the LP relaxation.

\paragraph{Solving the subproblem}
The problem of finding variables with negative reduced costs is almost equivalent to finding a maximum weight independent set in a circle graph, which can be solved in polynomial time by dynamic programming \cite{V2003,BB2012}.  A circle graph is an intersection graph of chords of a
a circle, two vertices/chords are adjacent if and only if they intersect. Circle graphs can be equivalently defined as the overlap graph of a set of intervals. The additional constraint that we impose is that there in the solution there are never more than $h$ intersecting intervals, corresponding to the height constraint. However, the algorithm to find the maximum weight independent set by dynamic programming can be easily adapted to take this constraint into account.

We will shortly describe how the problem of finding a sequence of moves with negative reduced costs for a fixed stack $s$ can be cast as finding a non-overlapping set of (labeled) intervals. For ease of exposition ignore the conditions on the initial and target lay-out of a stack and assume that all endpoints of the intervals are distinct. Let an interval $[a,b]$ with label $\ell$ mean that a container with priority $\ell$ is added to stack $s$ at time $a$ and removed at time $b$. Having two overlapping intervals $[a,b]$ and $[c,d]$ such that $a < c < b < d$ is interpreted as putting a container down at time $a$, putting another container on top at time $c$ and removing the first container at time $b$ {\em while the second container is still there}. Clearly this is infeasible. However, if $a < c < d < b$, then the second container would be put on top of the first, just as before, but {\em it would be removed before the first has to be removed}. Therefore, given a non-overlapping set of intervals whose endpoints are distinct, we have found a feasible sequence of moves.

\section{Branch and price algorithm} \label{sec:bnp}
The premarshalling problem is solved by iteratively running the algorithm with an increasing number of time points, starting from some lower bound, until the optimal solution is found. See Algorithm \ref{alg:bep} for an overview of the procedure.

\paragraph{Lower bound}
We say that a container is \emph{wrongly} placed if it is positioned on top of a container that (a) has a higher priority, or (b) is itself wrongly placed. Clearly, all wrongly placed containers need to be moved to obtain a target lay-out. However, if all stacks contain wrongly placed containers, then moving one does not reduce the number of wrongly placed containers. This number can only be reduced if at least one stack does not contain wrongly placed containers. The minimum number of moves required for this is equal to the lowest number of wrongly placed containers over all stacks. As lower bound we take the number of wrongly placed containers plus the minimum number of wrongly placed containers over all stacks.

\paragraph{Solving nodes}
For each node in a tree, $T$ time points are available to move containers. For solving a single node, we start with a model that contains (if any) previously generated sequences, and for each stack a dummy sequence with cost $T+1$, that performs no moves. These dummy sequences ensure that a feasible solution for the LP relaxation always exists. This initial model is solved, and as long as there are sequences with negative reduced cost, they are added and the model is resolved. Note that at each iteration at most one sequence is added per stack. If there are no more sequences with negative reduced cost, and the LP value exceeds $T$, we discard the node. If the LP value does not exceed $T$, we check if the solution is integral. If it is integral, we have found the optimal solution, and we stop the solve procedure. Otherwise, we apply the branching rule and continue with the next node. Every $100$ nodes the sequence pool is cleaned. All sequences that have not been used since the last cleanup, \ie, whose corresponding variable had value zero in the LP, are discarded.

\paragraph{Branching and node selection rule}
Observe that for a combination of stack and time point three actions are possible. Either (a) a container is added, (b) a container is removed, or (c) no move is performed. Let $\hat{t}$ be the minimum time point such that at least one of adding or removing a container is not forced for any stack at time $\hat{t}$. Out of all the stacks for which no action is forced at time $\hat{t}$, we take the stack $s$ for which the sum of all $x_{s,L}$ variables, such that for sequence $L$ a move is performed at time $\hat{t}$, is maximized, \ie, for which $\sum_{L \in \cL_s}\sum_{\ell \in [k]} (\addcontainer{L,\ell,\hat{t}}+\remcontainer{L,\ell,\hat{t}})x_{s,L}$ is maximized. In case of a tie, we take the stack $s$ that was considered first. By appropriately removing intervals, or adapting their value, this rule does not affect the difficulty of applying the separation oracle described in Section \ref{sec:ilp}. For exploring the tree we apply a depth first search. The node for which $\hat{t}$ is maximized, is considered next. In case of a tie, we take the most recently generated node.

\begin{algorithm}[h]
  \caption{Branch and Price algorithm\label{alg:bep}}
  \begin{algorithmic}[1]
    \State \textbf{procedure} \textsc{Premarshal}
      \State \quad set $T$ equal to the lower bound for the number of moves
      \State \quad \textbf{while} optimal solution not found \textbf{do}
        \State \quad\quad start with tree consisting of only a root node
        \State \quad\quad \textbf{while} exist unpruned leaf node \textbf{do}
          \State \quad\quad\quad $N \dfn$ leaf node deepest in tree, initialize LP model with dummy and ``valid'' 
          \State \quad\quad\quad sequences, solve LP model with column generation, update sequence pool
          \State \quad\quad\quad \textbf{if} $\text{LP value} > T$ \textbf{then} prune node $N$
          \State \quad\quad\quad \textbf{else} \textbf{if} solution integral \textbf{then} output opt. solution \& prune all nodes
          \State \quad\quad\quad \textbf{else} let $(s,t)$ denote the stack and time to branch on, add children
		  \State \quad\quad\quad\quad\quad\quad $N_{1}/N_{2}/N_{3} \dfn N \text{ with Add / Remove / Nothing fixed for } (s,t)$
          \State \quad\quad\quad every $100$ nodes clean sequence pool
        \State \quad\quad $T \dfn T+1$
  \end{algorithmic}
\end{algorithm}

\section{Experimental results}\label{sec:expRes}
In this section we evaluate our branch and price algorithm as described in Section \ref{sec:bnp}. We impose a time limit of one hour for each instance, and we only consider results for \prioritystacking. 

\subsection{Experimental setup}\label{sec:expRes_expSet}
The algorithm is implemented in \Cpp in combination with CPLEX 12.6, run on a machine with an Intel Core 2 Duo E8400 3.00 GHz processor and 4GB RAM, and evaluated on randomly generated instances. To the best of our knowledge no library with real-life instances for the premarshalling problem exists, and randomly generated instances are also used in for instance \cite{BF2012, CV2009b}. Furthermore, information that is required to determine the priority of a container is often not available or inaccurate at the time of unloading \cite{BF2012}, so that it is almost impossible to determine an appropriate position for the container.

The instances depend on four input parameters: the number of different priorities, the number of stacks, the height of the stacks, and the fill grade. For possible values for the number of priorities we consider \cite{LH2007}. To the best of our knowledge, this is the only other paper that considers an exact algorithm for the premarshalling problem. The authors basically consider two instances, with $3$ and $6$ priority levels, respectively. Therefore we consider at least $2$ (the minimum value possible) and at most $6$ different priorities. For the other parameters, Lee and Chao \cite{LC2009} observe that $12$ stacks with a height of $6$ is already larger than most equipment can handle, and a fill grade of $75\%$ is considered moderately high. A minimum height of $4$ is observed in general. As we apply an exact method, we also consider slightly lower parameter values, and thus smaller instances. For the number of stacks we take a maximum value of $9$, for the height of the stacks we take either $4$ or $6$, and for the fill grade we take either $50\%$ or $70\%$, which we consider a low and average fill grade, respectively.

For the exact values of these parameters, which are called \emph{Priorities}, \emph{Stacks}, \emph{Height}, and \emph{Fill}, respectively, see Table \ref{tab:parameterValues}. The number of containers is determined by multiplying the number of available positions, \ie, Stacks times Height, with Fill. In case this number is fractional, it is truncated. 

\begin{table}[h]
\centering
\begin{tabular}{C{1.8cm}|C{1.8cm}|C{1.8cm}|C{1.8cm}}
Priorities & Stacks & Height & Fill (\%) \\
\hline
2, 3, 6 & 3, 5, 7, 9 & 4, 6 & 50, 70 \\
\end{tabular}
\caption{Possible parameter values}\label{tab:parameterValues}
\end{table}

First, consider the case where Priorities has value $6$. Consider the containers one by one, and consecutively assign them priority $1$ to $6$. The initial lay-out is determined by randomly picking a container and placing it on a randomly selected non-full stack, until all containers are placed.

Second, consider the case where Priorities has value $2$ or $3$. In this case, the instances are based on the ones where Priorities has value $6$. Each stack has the same number of containers, but the priorities are updated. For Priorities equal to $2$, the three lowest and the three highest priorities are grouped together. For Priorities equal to $3$ the lowest two, middle two, and highest two priorities are grouped together.

If for any of the three values for Priorities the instance does not contain a wrongly placed container, \ie, no premarshalling operations are necessary, all three instances are discarded. This procedure is repeated until $20$ instances are generated for all $48$ combinations of parameter values, resulting in $960$ instances.

\subsection{Results}\label{sec:expRes_res}
Over all $960$ instances, the average mis-overlay is $38.3\%$. Clearly, the higher the value for Priorities, Height, and Fill, the higher the mis-overlay. For Stacks the value of mis-overlay is constant. This also follows from the way the instances are generated: the number of containers is (almost) linear in the number of stacks, and for each container a random stack is selected.

Out of all the instances, $945$ $(98.4\%)$ are solved within one hour, $895$ $(93.2\%)$ within one minute, and $680$ $(70.8\%)$ within one second. The $15$ instances that are not solved within one hour all have value $6$ for Priorities, value $6$ for Height, and value $70$ for Fill. The number of unsolved instances is $1$, $4$, $4$, and $6$ for $3$, $5$, $7$, and $9$ stacks, respectively.

The average running time is $93.65$ seconds, and is increasing in all four parameters. Especially for Priorities, Height, and Fill there is a big increase, which can be contributed to the unsolved instances, which have a running time of $3600$ seconds. Even if only the solved instances are considered, we still observe a substantial difference. The average time spend on solving the LP relaxation, generating columns, and clearing the cplex model is $46.66$ ($49.8\%$), $34.51$ ($36.8\%$), and $11.67$ seconds ($12.5\%$), respectively. This leaves $0.81$ seconds ($0.9\%$) of overhead. The percentage of time spend on solving the LP relaxation is increasing in Priorities, Height, and Fill. If we compare the solved and unsolved instances, we observe little difference in the relative amount spend on solving the LP relaxation ($47.5\%$ versus $51.4\%$), generating columns ($40.0\%$ versus $34.7\%$), and clearing the cplex model ($11.5\%$ versus $13.1\%$). 

For the solved instances the integrality gap is on average $1.13$. The integrality gap is obtained by comparing the optimal (integer) solution with the value of the LP relaxation at the root node of the tree that contains the optimal solution. This LP value is the lowest one obtained, as adding forced moves and time points respectively increases and decreases this value. The integrality gap is increasing in Priorities, Height, and Fill, decreasing in Stacks, and the maximum integrality gap is $2.80$. For the $15$ unsolved instances we cannot determine the integrality gap, but we can give a lower bound. For these instances the average and maximum lower bound on the integrality gap is $1.18$ and $1.85$, respectively. Although these values are biased, there does not appear to be a big difference in integrality gap between the solved and unsolved instances.

Consider the solved instances. On average $2.41$ trees are solved for these instances. Out of these trees, on average $0.80$ are \emph{killed} immediately. This occurs if the LP value of the root node exceeds the number of time points. As only one node is solved for killed trees, solving these trees generally requires little time. Hence, on average $1.61$ trees are \emph{actually} solved per instance. On average, each container is moved $0.49$ times. Just as for the integrality gap, the number of actually solved trees and the number of moves per container is increasing in Priorities, Height, and Fill, and decreasing in Stacks. The reason the number of moves per container is decreasing in Stacks is that with more stacks there are more options for temporarily storing containers, so that wrongly placed containers are handled less often.

If we consider the instances that cannot be solved, we observe the same behavior for the number of solved trees, and moves per container. With respect to the values, we see an increase. The average number of moves per container increases from $0.49$ to $0.83$. This increase is logical, as the unsolved instances exhibit a larger mis-overlay ($63.4\%$, compared to $38.3\%$ over all instances) so that most likely more moves are necessary to obtain a target lay-out. For Stacks equal to $3$ the number of moves is $1.58$, however, there is only a single instance for Stacks equal to $3$ that cannot be solved, so this value may not be representative. For the number of actually solved trees we also observe a substantial increase, from $1.61$ to $2.94$.

Over all instances, the maximum number of solved trees, killed trees, and actually solved trees is $11$, $4$, and $8$, respectively. Finally, we consider the memory usage. Therefore, we look at the nodes and generated sequences.

The average number of solved nodes is $295.7$. On average $0.32$ seconds are needed to solve one node. This average time is higher for the unsolved instances compared to the solved instances ($0.23$ versus $0.43$ seconds). The time needed per node is increasing in Priorities, Height, and Fill. Somewhat surprising is the fact that the longest time is obtained for Stacks equal to $5$ ($0.44$ seconds, compared to $0.26$, $0.30$, and $0.30$ seconds for $3$, $7$, and $9$ stacks, respectively). The average and maximum number of nodes in memory is $4.8$ and $56$, respectively. The maximum number of nodes in memory is obtained by a solved instance. For the unsolved instances, the maximum is $39$. Hence, the memory usage with respect to the number of nodes is low and stable over time, \ie, longer running times do not lead to a huge increase in the number of nodes in the memory. The average number of solved nodes and the average and maximum number of nodes in memory are all increasing in all four parameters.

The average and maximum number of generated sequences is $14,396$ and $1,491,758$, respectively. However, since every $100$ nodes unused sequences are discarded, maximally $59,729$ sequences are stored in the memory. In this case, the maximum is obtained by an unsolved instance, but the maximum over the solved instances is $56,554$. Hence, there again does not appear to be a huge increase in memory usage with increasing running time. A surprising result is that the most sequences are generated (and also stored in memory) for Stacks equal to $3$. In fact, if we consider the generated (stored) sequences per stack, we observe a decrease with respect to the number of stacks. Not all sequences that are stored in the memory are added to the LP model because of branching decisions. The maximum number of variables added to the LP model is $39,303$. Because unused sequences are discarded, the number of sequences stored in the memory and used in the LP model are kept at acceptable levels, at the expense of possibly generating the same sequence several times.

See Tables \ref{results_overview_1} and \ref{results_overview_2} for an overview. For both tables the rows represent the number of instances, the mis-overlay in percentage, the integrality gap, the number of trees solved, killed, and actually solved, the run time in seconds, the number of nodes solved and in memory, and the number of generated sequences.

\begin{table}[h]
\centering
\begin{tabular}{L{3.0cm}|C{2.7cm}|C{2.7cm}|C{2.7cm}}
statistic & all instances & solved instances & unsolved instances \\
\hline
\# instances & 960 & 945 & 15 \\
mis-overlay (\%) & 38.3 & 37.9 & 63.4 \\
int. gap & - & 1.13 & - \\
trees solved & 2.45 & 2.41 & 5.07 \\
trees killed & 0.82 & 0.80 & 2.13 \\
actually solved & 1.63 & 1.61 & 2.94 \\
run time (sec.) & 93.65 & 37.98 & 3600.84 \\
nodes solved & 295.7 & 168.6 & 8301.2 \\
nodes memory & 4.78 & 4.61 & 15.50 \\
sequences generated & 14,396 & 8,067 & 413,156 \\
\end{tabular}
\caption{Results split into categories all, solved, and unsolved instances.}
\label{results_overview_1}
\end{table}

\begin{table}[h]
\centering
\begin{tabular}{L{3.0cm}|C{2.0cm}|C{2.0cm}|C{2.0cm}|C{2.0cm}}
statistic & $<$ 1 sec & 1 sec - 1 min & 1 min - 1 h & $>$ 1 h \\
\hline
\# instances & 680 & 215 & 50 & 15 \\
mis-overlay (\%) & 32.3 & 50.7 & 58.9 & 63.4 \\
int. gap & 1.09 & 1.19 & 1.32 & - \\
trees solved & 1.90 & 3.40 & 5.12 & 5.07 \\
trees killed & 0.54 & 1.41 & 1.72 & 2.13 \\
actually solved & 1.36 & 1.99 & 3.40 & 2.94 \\
run time (sec.) & 0.23 & 8.42 & 678.39 & 3600.84 \\
nodes solved & 8.4 & 117.0 & 2569.2 & 8301.2 \\
nodes memory & 2.83 & 8.67 & 11.30 & 15.50 \\
sequences generated & 271 & 4,240 & 130,542 & 413,156 \\
\end{tabular}
\caption{Results split into categories running time below 1 second, between 1 second and 1 minute, between 1 minute and 1 hour, and more than 1 hour.}
\label{results_overview_2}
\end{table}

\section{Complexity: NP-hard for constant height stacks}\label{sec:nphard}
Because of limited space, we only state our results that both stacking problems are NP-hard for constant height stacks. We would like to point out
that the proof for \configurationstacking is significantly more involved and is not implied by a proof for \prioritystacking.

\begin{theorem}
For every fixed $h \geq 6$, \prioritystacking and \configurationstacking are NP-hard.
\end{theorem}

\section{Conclusion}
\label{sec:conclusion}
We considered the intra-bay premarshalling problem. The objective is to transform the initial lay-out into a target lay-out in the minimum number of moves possible. We showed that the premarshalling problem is NP-hard, even for a fixed stack height of at least six. We developed an exact algorithm based on branch and price, that is evaluated on $960$ randomly generated instances. For \prioritystacking, $945$ are solved within one hour, $895$ within one minute, and $680$ with one second. Preliminary experiments show that the algorithm runs much faster for \configurationstacking, in the full version there will be a comparison.
An interesting topic for future research concerns the solution approach. Currently, either the optimal solution is found, or no solution is found at all. By for instance incrementing the number of time points $T$ (see Section \ref{sec:bnp}) with more than one, it might be possible to look for other feasible (near-optimal) solutions. 

\bibliographystyle{plain}
\bibliography{stacking_2}

\newpage
\appendix

\section{Finding columns with negative reduced costs}
For clarity we repeat the definition of the integer linear program as described in Section \ref{sec:ilp}. Let the tuple $(\ell,t)$ denote a move of container $\ell$ at time $t$, and consider stack $s$. Let $\Ldown_{s}$ and $\Lup_{s}$ contain moves $(\ell,t)$ such that container $\ell$ is respectively added to, or removed from, stack $s$ at time $t$, and let $L_{s} \dfn ( \Ldown_{s} , \Lup_{s})$. The set $L_{s}$ is \emph{feasible} if (a) at every time $t$ there is at most one move; (b) for all moves $(\ell,t) \in \Ldown_{s}$ and $(\ell,t) \in \Lup_{s}$ we have that at time $t$ stack $s$ contains at least one free spot, or container $\ell$ is the top container of stack $s$, respectively; (c) the lay-out obtained by executing the moves is a target lay-out. Hence, a feasible $L_{s}$ can be viewed as a sequence of moves that transforms stack $s$ into a target lay-out, where $\Ldown_{s}$ and $\Lup_{s}$ consist of the moves where a container is added or removed, respectively. Furthermore, let $\nummoves{L_{s}}$ denote the number of containers added in $L_{s}$, \ie, $\nummoves{L_{s}} \dfn |\Ldown_{s}|$. Finally, Let $\cL_{s}$ denote the set of sequences that transform stack $s$ into a target lay-out, \ie, $\cL_{s} \dfn \{L_{s} : L_{s} \text{ is feasible}\}$.

We consider the following ILP, where $x_{s,L} \in \{0,1\}$ for $L \in \cL_{s}$ has value $1$ if and only if stack $s$ is transformed according to sequence $L$. Let $\addcontainer{L,\ell,t}$ and $\remcontainer{L,\ell,t}$ be equal to $1$ if and only if $(\ell,t) \in \Ldown$ and $(\ell,t) \in \Lup$, respectively.

\begin{alignat}{3} \tag{$\mathsf{Integer~Linear~Program}$}
\text{min} \quad & \sum_{s \in [m]} \sum_{L \in \cL_s} \nummoves{L} \ x_{s,L}  & \\
\text{s.t.} \quad & \sum_{s \in [m]} \sum_{L \in \cL_s} \left( \addcontainer{L,\ell,t}-\remcontainer{L,\ell,t}\right) x_{s,L}  \geq 0 & \quad \forall \ell \in [k], t \in [T] \tag{C1}  \\
& \sum_{s \in [m]} \sum_{L \in \cL_s}\sum_{\ell \in [k]} \addcontainer{L,\ell,t} x_{s,L} \leq 1 & \forall t \in [T] \tag{C2}  \\
& \sum_{L \in \cL_s} x_{s,L} = 1 & \forall s \in [m] \tag{C3}  \\
& x_{s,L} \in \{0,1\} &  \quad \forall s \in [m], L \in \cL_s \notag
\end{alignat}

By relaxing the requirement that the variables are either $0$ or $1$ we obtain the relaxation $\mathsf{LP~Primal}$.

\begin{alignat}{3} \tag{$\mathsf{LP~Primal}$}\label{LP_primal}
\text{min} \quad & \sum_{s \in [m]} \sum_{L \in \cL_s} \nummoves{L} \ x_{s,L}  & \\
\text{s.t.} \quad & \text{Constraints } \eqref{addRemove},\eqref{oneMove}, \text{ and } \eqref{oneSelected} & \notag \\
& x_{s,L} \geq 0 & \quad \forall s \in [m], L \in \cL_s \notag
\end{alignat}

Let $\alpha_{\ell,t}$, $\beta_{t}$, and $\gamma_{s}$ be the dual variables for constraints of type \eqref{addRemove}, \eqref{oneMove} and \eqref{oneSelected}, respectively. The dual LP is as follows.

\begin{alignat}{3} \tag{$\mathsf{LP~Dual}$}\label{LP_dual}
\text{max} \quad & \sum_{t \in [T]} -\beta_{t} + \sum_{s \in [m]} \gamma_{s} & \\
\text{s.t.} \quad & \sum_{\ell \in [k]}\sum_{t \in [T]} \left(\addcontainer{L,\ell,t} - \remcontainer{L,\ell,t}\right) \alpha_{\ell,t} - \notag \\
& \quad\quad \sum_{\ell \in [k]}\sum_{t \in [T]} \addcontainer{L,\ell,t} \beta_{t} + \gamma_{s} \leq \nummoves{L} & \quad \forall s \in [m], L \in \cL_s \tag{D1} \label{dualConstraint} \\
& \alpha_{\ell,t}  \geq 0 & \quad \forall \ell \in [k], t \in [T] \notag \\
& \beta_{t}  \geq 0 & \quad \forall t \in [T] \notag
\end{alignat}

For column generation we need an algorithm that finds variables with negative reduced costs. Given the dual variables $\alpha_{\ell,t}$, $\beta_{t}$, and $\gamma_{s}$, the goal is to find (if any) $L \in \cL_{s}$ with negative reduced costs, \ie, for a specific stack $s$ the goal is to maximize
\[
\sum_{\ell \in [k]}\sum_{t \in [T]} \left(\addcontainer{L,\ell,t}\alpha_{\ell,t} - \remcontainer{L,\ell,t}\alpha_{\ell,t} - \addcontainer{L,\ell,t} \beta_{t}\right) + \gamma_{s} -  \nummoves{L} 
\]

Consider an arbitrary stack $s$. Since $\gamma_{s}$ is a constant for stack $s$, we disregard it for now. Without loss of generality assume that the initial lay-out is of height $h$ and let $X \dfn \{x_1,...,x_h\}$ denote this initial lay-out. We will phrase the goal of finding the maximum weight sequence for this stack as finding the maximum weight set of intervals $I$ such that
\begin{enumerate}
 \item[(a)] no two intervals start or end at the same time;
 \item[(b)] all intervals are non-overlapping;
 \item[(c)] there is no point $t$ such that more than $h$ intervals from $I$ intersect with $t$, \ie, the height is at most $h$.
\end{enumerate}

Without condition (a) this would be equivalent to finding a maximum weight independent set in the overlap graph or circle graph defined by the intervals. In fact, by shifting the intervals slightly we can cast the question as such but at the cost of a tedious description of the instance. Therefore, we include condition (a) which will simplify the construction. A set of intervals that adheres to these conditions is called \emph{feasible}. The start and end point of an interval coincide with placing the container on the stack, and removing the container from the stack, respectively.

We consider time points $-h+1,\ldots,T+k$, where time points $-h+1,\ldots,0$ are used to represent the initial lay-out, and time point $T+1,\ldots,T+k$ to represent the final lay-out. Time points $1,\ldots,T$ are the regular time points, where one container can be moved per time point. Hence, for intervals whose start point is in the range $-h+1,\ldots,0$, the corresponding container is not moved to the stack but is rather part of the initial lay-out. Similarly, for intervals whose end point is in the range $T+1,\ldots,T+k$, the corresponding container is not removed from the stack but is rather part of the final lay-out. Using these time points, we define the following intervals $I$ and weights $w$, where $B$ is equal to the sum of all $\alpha_{\ell,t}$ and $\beta_{t}$ values.

\begin{itemize}
  \item {\em Containers with label $\ell$ that are initially at height $i$, and which are not moved:}\\
        $[i{-}h, T{+}\ell; \ell]$ with weight $w[i{-}h, T{+}\ell; \ell] = 0 + 2B$
  \item {\em Containers with label $\ell$ that are initially at height $i$, but are removed:}\\
        $[i{-}h, t; \ell]$ for all $1 \leq t \leq T$ with weight $w[i{-}h, t; \ell] = -\alpha_{\ell,t} + 2B$;
  \item {\em Containers with label $\ell$ that are placed on the stack and never removed:} \\
        $[s, T{+}\ell; \ell]$ for all $1 \leq s \leq T$ with weight $w[s, T{+}\ell; \ell] = \alpha_{\ell,s}-\beta_{s}-1$;
  \item {\em Containers with label $\ell$ that are placed and later removed from the stack:}\\
        $[s,t; \ell]$ for all $1 \leq s < t \leq T$ with weight $w[s,t;\ell] = \alpha_{\ell,s}-\alpha_{\ell,t}-\beta_{s}-1$.
\end{itemize}

Intervals with an added weight of $2B$ are called \emph{heavy} intervals. For an illustration of these intervals, see Figure \ref{fig:interval}. 

\begin{figure}
\centering
\includegraphics[width=0.9\textwidth]{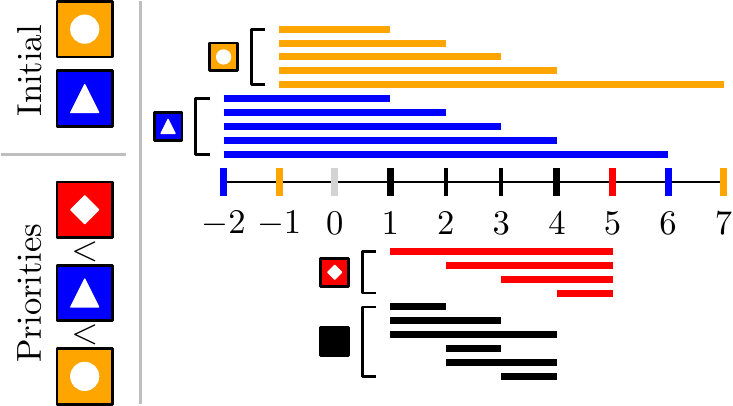}\\
\caption{This figure shows on the left a small instance consisting of three types of containers, called ``red'', ``blue'', and ``orange'', which can also be identified by the diamond, triangle, or circle, respectively. The priorities for these containers are $1$, $2$, and $3$, respectively. On the right the corresponding intervals for $T=4$ moves and maximum height $h=3$ are depicted. Above the time line the intervals specific for orange and blue are shown. Below the time line the intervals specific for red and the intervals pertaining to all three types are depicted. Note that the intervals above the line are exactly the heavy intervals.}
\label{fig:interval}
\end{figure}

\begin{lemma}
The maximum weight feasible set of intervals represents a maximum weight sequence.
\end{lemma}

\begin{proof}
Consider containers $x_{i}$ and $x_{i+1}$ in the initial lay-out. All intervals pertaining to $x_{i}$ and $x_{i+1}$ start at time point $i-h$ and $i-h+1$, respectively. Hence, the start point for $x_{i}$ is to the ``left'' of the start point for $x_{i+1}$. By the previously mentioned conditions, we must have that the end point of the interval for $x_{i}$ is to the ``right'' of the end point for $x_{i+1}$. This implies that $x_{i}$ can only be removed after $x_{i+1}$, which is indeed correct, as $x_{i}$ is placed below $x_{i+1}$. 

Disregarding the added weight of $2B$ for the intervals pertaining to the initial solution, we have that any feasible set of intervals has a weight ranging between $-B$ and $B$. Hence, adding a heavy interval at the expense of other intervals always increases the weight of the interval. Since all heavy intervals start in the range $-h+1,\ldots,0$ and at most one selected interval may start per time point, at most $h$ heavy intervals can be selected. Combining these observations, exactly $h$ heavy intervals will be selected, so that for every container in the initial lay-out a corresponding interval is selected.

Likewise, containers with higher priority are placed above containers with lower priority in the final lay-out. To see this, consider two containers $\ell$ and $\ell+1$, which have end point $T+\ell$ and $T+\ell+1$, respectively. Again by the conditions imposed on the selected intervals, we must have that the start point for $\ell+1$ is to the ``left'' of the start point for $\ell$. This implies that $\ell+1$ is placed first, so that $\ell$ is placed above $\ell+1$, which corresponds to a desired final lay-out. 
\end{proof}

Now we can set up a dynamic program to find the maximum weight feasible set of intervals. Let $I[a,b]$ contain the subset of intervals that is a subset of the range $[a,b]$, \ie, $I[a,b] \dfn \{[c,d;\ell] \in I:[c,d] \subseteq [a,b]\}$. For any interval $[a,b]$ and remaining height $j$ we consider all ``leftmost'' intervals $[a,c]$ for $a < c \leq b$. A leftmost interval either covers no container, which implies that no containers are added to or removed from the stack in this interval, or it covers a container $\ell$, which implies that other containers can be added to or removed from the stack in the interval $[a+1,c-1]$. Clearly, containers can always be added to or removed from the stack in the remaining rightmost interval.

Consider states $A[a,b,j]$ for $a,b \in \{-h+1,\ldots,T+k\}$ and $j \in \{0,\ldots,h\}$, which represents the optimal solution restricted to interval $[a,b]$ and remaining height $j$. By definition, we have that $A[a,b,j] \dfn 0$ if $b \leq a$ or $j = 0$ or $|\{ i \in I : i \subseteq [a,b]\}| = 0$. For all other states we have the following recursion

\[
A[a,b,j] \dfn \max_{a < c \leq b}
\begin{cases}
A[c,b,j] \\
\displaystyle\max_{i=[a,c;\ell] \in I[a,c]} \left\{w(i) + A[a{+}1,c{-}1,j{-}1] + A[c{+}1,b,j]\right\}
\end{cases}
\]

The running time is clearly polynomial, since evaluating a single entry in the dynamic program takes at most $O((T{+}h{+}k)k)$ time and there are $(T{+}h{+}k)^{2}(h{+}1)$ entries.

The approach for finding violated constraints for \configurationstacking is basically identical to the one for \prioritystacking. The only difference is that we now want to obtain a specific target lay-out. This is achieved in the same manner as that the initial lay-out is enforced. Hence, in stead of time points $T+1,\ldots,T+k$ for the final lay-out, we now use time points $T+1,\ldots,T+h$. Let $Y \dfn \{y_1,...,y_h\}$ denote the target lay-out. Intervals corresponding to container $y_{i}$ have $T+h-i+1$ as end point. Using these time points, we now have the following intervals $I$ and weights.

\begin{itemize}
  \item {\em Containers with label $\ell$ that are at height $i$ in the initial and target lay-out, and which are not moved:}\\
        $[i{-}h, T{+}h{-}i{+}1; \ell]$ with weight $w[i{-}h, T{+}h{-}i{+}1; \ell] = 0 + 4B$
  \item {\em Containers with label $\ell$ that are initially at height $i$, but are removed:}\\
        $[i{-}h, t; \ell]$ for all $1 \leq t \leq T$ with weight $w[i{-}h, t; \ell] = -\alpha_{\ell,t} + 2B$;
  \item {\em Containers with label $\ell$ that are placed on the stack and never removed, and are at height $i$ in the target lay-out:} \\
        $[s, T{+}h{-}i{+}1; \ell]$ for all $1 \leq s \leq T$ with weight $w[s, T{+}h{-}i{+}1; \ell] = \alpha_{\ell,s}-\beta_{s}-1 +2B$;
  \item {\em Containers with label $\ell$ that are placed and later removed from the stack:}\\
        $[s,t; \ell]$ for all $1 \leq s < t \leq T$ with weight $w[s,t;\ell] = \alpha_{\ell,s}-\alpha_{\ell,t}-\beta_{s}-1$.
\end{itemize}

\section{Complexity} \label{sec:complexity}
In this section, we review the complexity status of \prioritystacking and \configurationstacking. All reductions are from the Mutual Exclusion Scheduling problem on permutation graphs, see Jansen \cite{Jansen03}. Given a permutation $\pi:[n] \rightarrow [n]$, the permutation graph $G=([n],E)$ has a vertex for each number in $[n]$ and there is an edge $(i,j)\in E$ for all $i < j$ if and only if $\pi(i) > \pi(j)$. Permutation graphs are exactly the class of graphs that are both \emph{comparability graphs} and \emph{co-comparability graphs}. Permutation graphs can also be viewed as an \emph{intersection graph} of a set of line segments that connect two parallel lines.

For \exclusionscheduling we are given a permutation graph $G=([n],E)$, and integers $k \geq 1$ and $m \geq 2$. The goal is to partition the vertex set into disjoint sets $P_1,\ldots,P_k$ with $|P_i| \leq m$ and $P_i$ an independent set in $G$, for all $i \in [k]$. Notice that $P_i$ is independent if and only if for all $u,v \in P_i$ such that $u < v$ we have that $\pi(u) < \pi(v)$. Jansen proved that even for fixed $m \geq 6$ the problem is NP-hard for permutation graphs.

\begin{theorem}[\cite{Jansen03}]
For each fixed $m \geq 6$, \exclusionscheduling is NP-hard for permutation graphs.
\end{theorem}

We show that both \prioritystacking and \configurationstacking are NP-hard for fixed stack height. Both proofs are based on the same principle, which we will explain and elucidate by first proving that \prioritystacking is NP-hard for \emph{arbitrary} stack height. In these proofs, we will use three types of containers, called \emph{node}, \emph{dummy}, and \emph{fill} containers. For each node $i \in [n]$ of the permutation graph, there will be a node container $i$, which also has priority level $i$. All other containers will be either dummy or fill containers, where dummy containers have to be moved at least once, \ie, it is wrongly placed initially, while fill containers do not necessarily have to be moved. In figures, the containers are depicted by square boxes. Node and dummy containers can be identified by black or rounded corners, respectively.

\subsection{\prioritystacking is NP-hard for arbitrary height}
\label{sec:np_hard_priority_arbitrary}
We prove that \prioritystacking is NP-hard when the maximum stack height is not fixed. We construct an instance of \prioritystacking as follows. The maximum height $h$ is equal to $n+1$, and there are $k+1$ stacks. The initial lay-out for stacks $1,\ldots,k$ consist of $n-m+1$ fill containers with priority $\infty$. For stack $k+1$ the initial lay-out is $\{s, \pi(1), \ldots, \pi(n)\}$, where the priority of fill container $s$ is equal to $0$. See Figure~\ref{fig:3_npc_priority_arbitrary} for an illustration, where for each container the name is given and the fill containers with priority $\infty$ are depicted by black boxes.

\begin{figure}[h]
\centering
\includegraphics[scale=0.8]{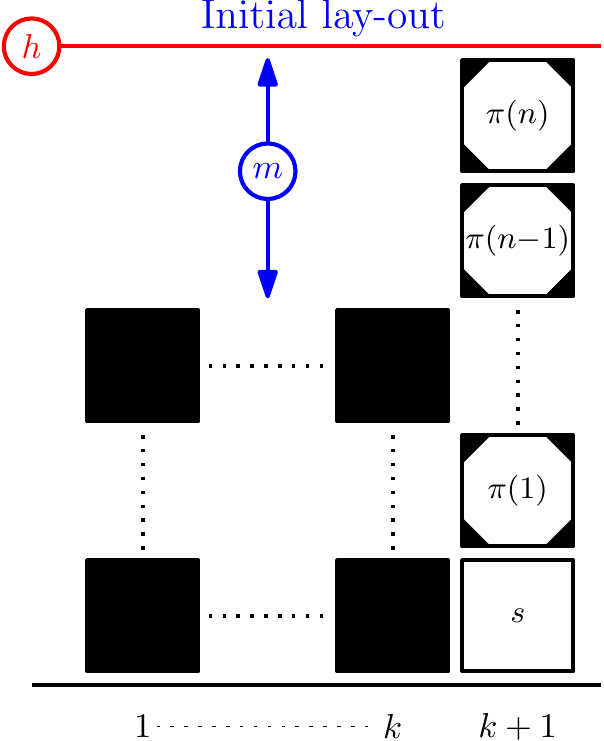}
\caption{\prioritystacking arbitrary height}
\label{fig:3_npc_priority_arbitrary}
\end{figure}

\begin{theorem}\label{thm:3_npc_priority_arbitrary}
\prioritystacking is NP-hard.
\end{theorem}
\begin{proof}
Observe that all node containers are placed on top of fill container $s$, which has a higher priority than all node containers. Hence, all node containers need to be moved at least once to obtain a target lay-out. Hence, the lower bound \lb on the number of moves for this instance is $n$. We show that a solution for \exclusionscheduling exists if and only if a solution of \lb moves exists for \prioritystacking.

Assume that there is a solution for \prioritystacking of \lb moves. Hence, all node containers are moved exactly once. Since they are initially placed on the same stack, there is only in order in which they can be moved, namely $\pi(n),\ldots,\pi(1)$. Furthermore, since the node containers are initially on stack $k+1$, they are all moved to stacks $1,\ldots,k$. Each of these stacks contain $n-m-1$ fill containers, so that at most $m$ node containers are moved to each of these stacks. Hence, the two following properties hold by definition.

\begin{property}
\label{prop_order}
The node containers are moved in the order $\pi(n),\ldots,\pi(1)$. 
\end{property}

\begin{property}
\label{prop_layout}
All node containers are on stacks $1,\ldots,k$ after they are moved. Each of these stacks contains at most $m$ node containers.
\end{property}

Using these two properties, we state the following lemma.

\begin{lemma}
\label{lem_mes}
After all node containers are moved, the lay-out of stacks $1,\ldots,k$ constitutes a solution for \exclusionscheduling
\end{lemma}

\begin{proof}
For $j \in [k]$, let $P_{j}$ contain the nodes $i$ such that node container $\pi(i)$ is placed on stack $j$. By Property \ref{prop_layout} we know that $P_{1},\ldots,P_{k}$ form a partition, and that $|P_{j}| \leq m$ for all $j$. Hence, we only need to show that each $P_{j}$ is an independent set in the permutation graph. 

Clearly, $P_{j}$ is an independent set if it consists of at most one element. Take an arbitrary $j$ such that $|P_{j}| \geq 2$, and let $\pi(a)$ and $\pi(b)$, with $a < b$, be two arbitrary node containers on stack $j$. Hence, we have that $\{a,b\} \subseteq P_{j}$. As $a < b$, we have that container $\pi(b)$ was moved to stack $j$ before container $\pi(a)$, so that $\pi(a)$ is placed above $\pi(b)$. In order to have a target lay-out, we need that the stack is sorted in non-increasing order when viewed from the bottom. This immediately implies that $\pi(a) < \pi(b)$, so that there is no edge between nodes $a$ and $b$ in the permutation graph. As the two nodes were chosen arbitrarily, it holds for all pairs of nodes, so that $P_{j}$ is an independent set. As $j$ was also chosen arbitrarily, it holds for all $P_{j}$, proving the lemma.
\end{proof}

Hence, by Lemma \ref{lem_mes} we also have a solution for \exclusionscheduling. Assume that there is a solution for \exclusionscheduling. We construct the following solution of \lb moves for \prioritystacking: move the node containers in order $\pi(n),\ldots,\pi(1)$, where container $\pi(i)$ is moved to stack $j$ such that $i \in P_{j}$. The validity of these moves follows from the initial lay-out of stack $k+1$, and the fact that at most $m$ containers are added to each stack $1,\ldots,k$. Observe that Property \ref{prop_order} holds for this sequence. We proof the following lemma.

\begin{lemma}
\label{lem_stacking}
After all node containers are moved, the node containers in stacks $1,\ldots,k$ are sorted in non-increasing order when viewed from the bottom.
\end{lemma}

\begin{proof}
Take an arbitrary stack $j$. If this stack contains at most one node container, then the claim trivially holds. So, furthermore assume that $j$ contains at least two node containers, and let $\pi(a)$ and $\pi(b)$, with $a<b$, be two arbitrary node containers in this stack. By construction, we have that $\{a,b\} \subseteq P_{j}$, which implies that $\pi(a) < \pi(b)$, as otherwise there would be an edge between $a$ and $b$. By Property \ref{prop_order} we also know that $\pi(b)$ was moved to stack $j$ before $\pi(a)$, so that $\pi(a)$ is placed above $\pi(b)$. This immediately implies that $\pi(a)$ and $\pi(b)$ are sorted in the correct order. As the two node containers were chosen arbitrarily, it holds for all pairs of node containers, so that the claim holds for stack $j$. As $j$ was also chosen arbitrarily, the claim holds for all stacks $1,\ldots,k$.
\end{proof}

By Lemma \ref{lem_stacking}, we also have that these stacks are sorted in non-increasing order, so that we also have a target lay-out. The constructed sequence of moves is thus a solution for \prioritystacking, so that we have that \prioritystacking is NP-hard for arbitrary stack heights.
\end{proof}

\subsection{\prioritystacking is NP-hard for fixed height}
In the previous section we showed that \prioritystacking is NP-hard for arbitrary stack height. In this section, we consider the case where the maximum stack height is fixed. We construct an instance of \prioritystacking as follows. The maximum height $h$ is equal to $m$. For each $j \in [k]$ there is a stack $j$ which is initially empty. For each $i \in [n]$ there is a stack $s_{1,i}$. Stack $s_{1,1}$ initially consists of fill and node container ${-}n$ and $\pi(n)$, respectively. For the other stacks $s_{1,i}$ the initial lay-out consists of fill, node, and dummy container ${-}n{+}i{-}1$, $\pi(n{-}i{+}1)$, and ${-}n{+}i{-}2$, respectively. Finally, there are $z=km-n$ stacks, called $s_{2,i}$ for $i=1,\ldots,z$, with as initial lay-out fill and dummy container ${-}\infty$ and $0$, respectively. See Figure~\ref{fig:4_npc_priority_fixed} for an illustration, where for each container the priority is given.

\begin{figure}[h]
\centering
\includegraphics[scale=0.8]{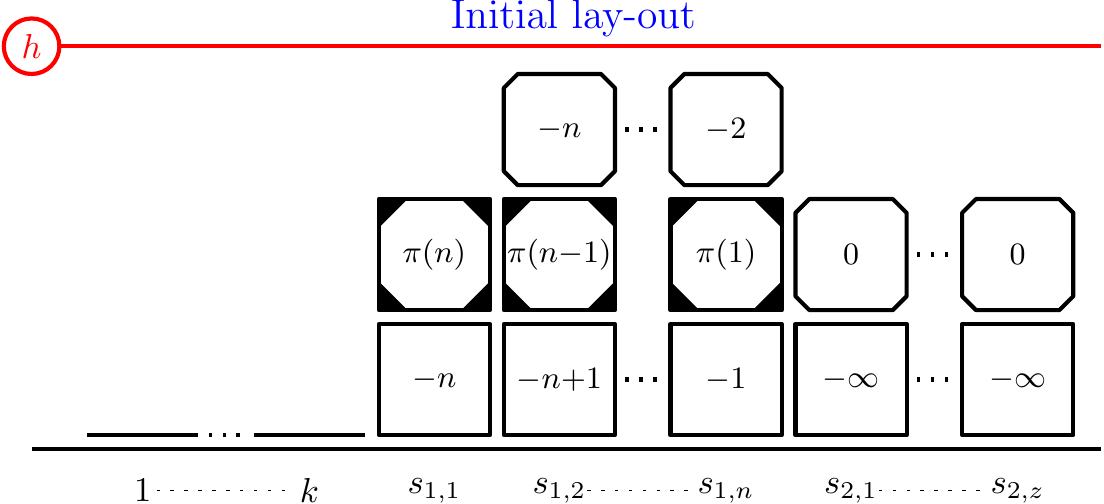}
\caption{\prioritystacking fixed height}
\label{fig:4_npc_priority_fixed}
\end{figure}

\begin{theorem}\label{thm:4_npc_priority_fixed}
For every fixed $h \geq 6$, \prioritystacking is NP-hard.
\end{theorem}

\begin{proof}
Consider Figure~\ref{fig:4_npc_priority_fixed}, and observe that every node container is placed on top of a fill container with a higher priority. Furthermore, all dummy containers are either placed on a node container, or on a fill container with a higher priority. Hence, all node and dummy containers need to be moved at least once, so that the lower bound \lb on the number of moves is equal to $z+2n-1$.

Assume that there is a solution for \prioritystacking of \lb moves. Observe that the node containers can only be moved to stacks $1,\ldots,k$, as all other stacks contain a fill container with a higher priority. Since the maximum stack height is $m$, by definition at most $m$ node containers can be moved to each of these stacks. This already shows that Property \ref{prop_layout} holds.

Consider the $z$ dummy containers with priority $0$. As for the $n$ node containers, these dummy containers can only be moved to stacks $1,\ldots,k$. Together, these node and dummy containers sum up to $km$, which coincided with the total number of spots for these stacks. This implies that the dummy containers with priority ${-}2,\ldots,{-}n$ cannot be moved to one of these stacks. These containers can also not be moved to stacks $s_{2,1},\ldots,s_{2,z}$, since they contain a fill container with priority ${-}\infty$. Hence, for $i=1,\ldots,n-1$, dummy container ${-}i{-}1$, which is initially placed on stack $s_{1,n-i+1}$, can only be moved to stacks $s_{1,n-i}$ and $s_{1,n-i+2},\ldots,s_{1,n}$, \ie, it can be moved one stack to the ``left'' or any stack to the ``right''.

Assume that Property \ref{prop_order} does \emph{not} hold, \ie, the node containers are not moved in order the $\pi(n),\ldots,\pi(1)$. Let $\pi(i)$ denote the node container with highest index $i$ that is moved before $\pi(i{+}1)$. From this we know that node containers $\pi(1),\ldots,\pi(i{-}1)$ and $\pi(i{+}1)$ are moved after $\pi(i)$, which in turn is moved after dummy container ${-}i{-}1$, since it is placed below this dummy container. But then, we have that dummy container ${-}i{-}1$ can only be moved to a stack that still contains a node container. That would block this node container from moving, preventing a target lay-out from being reached. Hence, the assumption can not hold, so that Property \ref{prop_order} indeed holds.

Consider Lemma \ref{lem_mes} from the previous section. Since Properties \ref{prop_order} and \ref{prop_layout} both hold for the current sequence of moves, the lemma also holds for this instance. Hence, we also have a solution for \exclusionscheduling. 

Assume that a solution for \exclusionscheduling exists. Consider the following sequence of moves for \prioritystacking. Start with moving node container $\pi(n)$ to stack $j$ such that $n \in P_{j}$. Next, for $i=2,\ldots,n$ repeat the following two steps: move dummy container $i{-}n{-}2$ from stack $s_{1,i}$ to $s_{1,i-1}$, \ie, move it one stack to the ``left'', then move node container $\pi(n{-}i{+}1)$ to stack $j$ such that $n{-}i{+}1 \in P_{j}$. Finally, for $i \in [z]$, move the dummy container on stack $s_{2,i}$ to the ``leftmost'' stack among $1,\ldots,k$ that is not full. Observe that this sequence satisfies Property \ref{prop_order}.

Note that the order in which the containers are moved is from left to right and then from top to bottom, and all containers are moved to the left. Also no stack will ever contain more than $m$ containers, so that the sequence of moves is valid.

Consider the final lay-out. Stacks $s_{1,1},\ldots,s_{2,z}$ contain either one container, or two containers with the same priority, which are both target lay-outs. For stacks $1,\ldots,k$, we have that first the node containers are placed, and then the dummy containers $0$, which have a higher priority. Hence, if the node containers in a stack are sorted, the entire stack is sorted. As the sequence of moves satisfies Property \ref{prop_order}, Lemma \ref{lem_stacking} from the previous section still holds. This implies that the lay-out for stacks $1,\ldots,k$ is also a target lay-out.

Hence, all stacks exhibit a target lay-out, so that we also have a solution for \prioritystacking. This implies that for every fixed $h \geq 6$ \prioritystacking is NP-hard.
\end{proof}

\subsection{\configurationstacking is NP-hard for fixed height}
\label{sec:np_hard_configuration_fixed}
In the previous sections we showed that \prioritystacking is NP-hard, both in the case of arbitrary and fixed stack height. In this section we show that \configurationstacking is NP-hard if the maximum height is a constant with value at least $6$.

We construct an instance of \configurationstacking as follows. The maximum height $h$ is equal to $m$ and there are $4n+k+2$ stacks, which can be split into five groups. For each vertex $i$ in the permutation graph there are dummy containers $a_{i}$ to $h_{i}$, except for container $h_{n}$. Furthermore there are dummy containers $x$ and $y$, and $z=(n+1)(h-3)$ fill container. All fill containers have priority $\infty$, while all other containers have a unique priority.

Consider the initial lay-out. Stacks $1,\ldots,k$ are initially empty. Stacks $s_{1,1},\ldots,s_{1,n}$ consist of five containers, where the lay-out of stack $s_{1,i}$ is $\{g_{\pi(n{-}i{+}1}, d_{i}, b_{i}, \pi(n{-}i{+}1), a_{i}\}$. Stack $s_{2,1}$ contains $e_{1}$, and for $i=2,\ldots,n$, stack $s_{2,i}$ consists of containers $h_{i-1}$ and $e_{i}$. Stacks $s_{3,0},\ldots,s_{3,n}$ contain $h-3$ fill containers each. For $i=1,\ldots,n-1$, stack $s_{3,i}$ additionally contains dummy container $c_{i}$, while stack $s_{3,n}$ additionally consists of containers $x$ and $c_{n}$. Stack $s_{4,0}$ is initially empty, and stack $s_{4,i}$ consists of $f_{i}$ for $i=1,\ldots,n-1$. Stack $s_{4,n}$ contains dummy containers $y$ and $f_{n}$.

Consider the target lay-out. Stacks $1,\ldots,k$ and $s_{1,1},\ldots,s_{1,n}$ are empty. The lay-out for stacks $s_{2,1}$ and $s_{2,n}$ is $\{y,g_{1},1,h_{1}\}$ and $\{g_{n},n\}$, respectively. Stack $s_{2,i}$ consists of containers $g_{i}$, $i$, and $h_{i}$, for $i=2,\ldots,n-1$. Stacks $s_{3,0},\ldots,s_{3,n}$ contain $h-3$ fill containers each. For $i=1,\ldots,n$, stack $s_{3,i-1}$ additionally contains dummy containers $a_{i}$, $b_{i}$, and $c_{i}$. Stack $s_{4,0}$ consists of $4$ containers, namely $x$, $d_{1}$, $e_{1}$, and $f_{1}$. For $i=2,\ldots,n$, the lay-out of stack $s_{4,i-1}$ is $\{d_{i}, e_{i}, f_{i}\}$. Finally, stack $s_{4,n}$ is empty in the target lay-out.

See Figure~\ref{fig:2_npc_configuration_fixed} for an illustration of the initial and target lay-out, where for each container the name is given and the fill containers are depicted by black boxes.

\begin{figure}[h]
\centering
\includegraphics[width=1\textwidth]{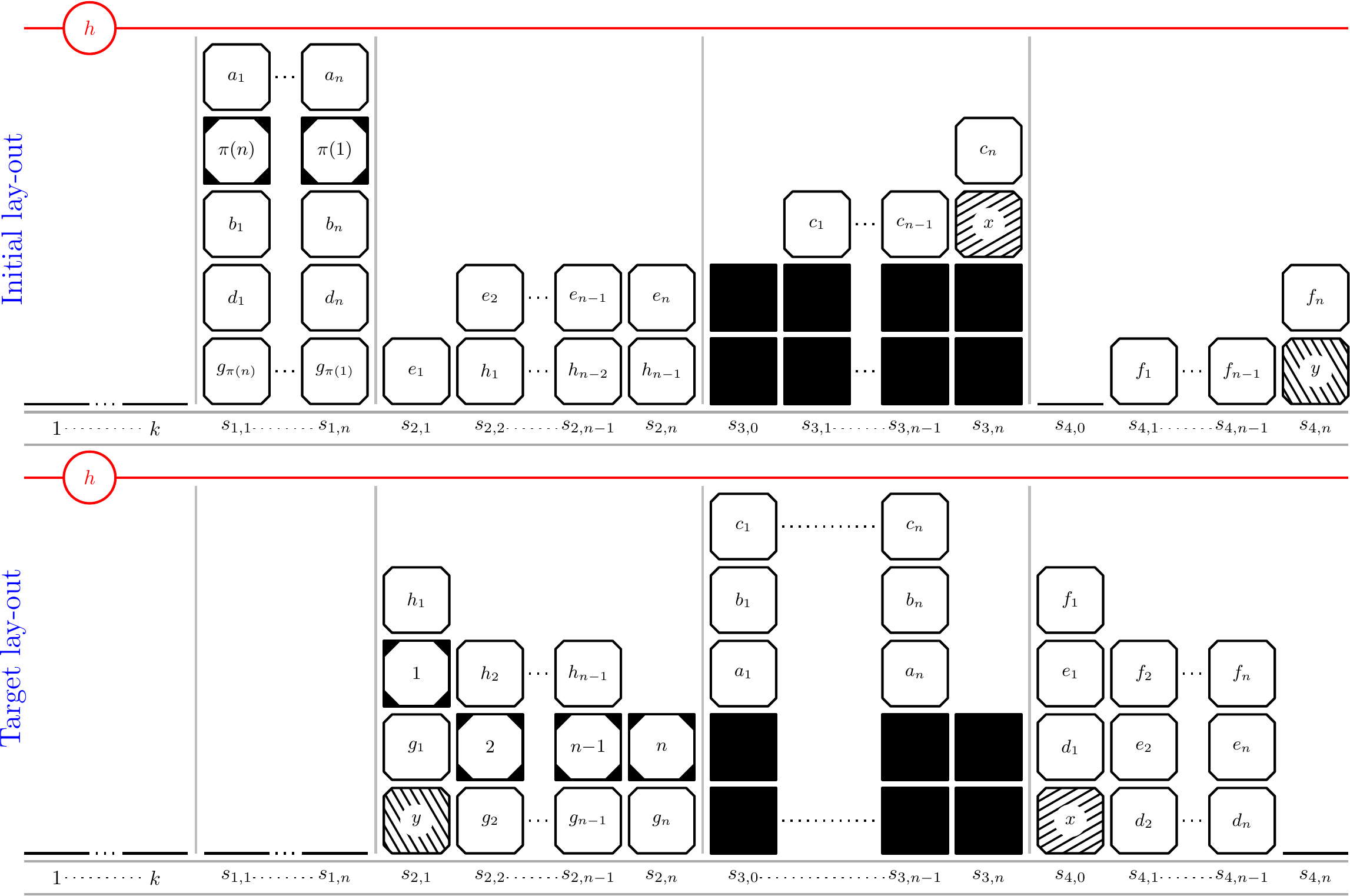}
\caption{\configurationstacking fixed height}
\label{fig:2_npc_configuration_fixed}
\end{figure}

\begin{theorem}\label{thm:2_npc_configuration_fixed}
For every fixed $H \geq 6$, \configurationstacking is NP-hard.
\end{theorem}
\begin{proof}
Take an arbitrary node container $i$, and observe that it is placed on top of dummy container $g_{i}$ in both the initial and target lay-out. This implies that $i$ has to be moved at least twice, once before $g_{i}$ can be moved, and once after $g_{i}$ is moved. As $i$ was chosen arbitrarily, this holds for all node containers. Furthermore, all dummy containers need to be moved at least once, so that the lower bound \lb on the number of moves for this instance is $10n+1$.

Assume that there is a solution for \configurationstacking of \lb moves. This implies that the node containers are moved twice and the dummy containers once. We claim that there is only a single order of moves that transforms the initial lay-out into the target lay-out in \lb moves. To see why this is true, let $t_{\alpha}$ denote the time point at which dummy container $\alpha$ is moved, and let $t_{i}^{(1)}$ and $t_{i}^{(2)}$ denote the first and second time point at which node container $i$ is moved. Consider the following observations.

\begin{observation}\label{obs_phase_1}
$t_{a_{i}} < t_{\pi(n-i+1)}^{(1)} < t_{b_{i}} < t_{c_{i}}$, $\forall i \in [n]$, and $t_{c_{i}} < t_{a_{i+1}}$, $\forall i \in [n{-}1]$.
\end{observation}

\begin{proof}
The first two inequalities follow from the fact $a_{i}$ is placed on top of $\pi(n{-}i{+}1)$, which in turn is placed on top of $b_{i}$, in the initial lay-out of stack $s_{1,i}$. The third inequality follows from the target lay-out of stack $s_{3,i-1}$, where $c_{i}$ is placed on top of $b_{i}$. The last inequality follows from the fact that $a_{i+1}$ has to be moved to the initial position of $c_{i}$. This is only possible if $c_{i}$ is moved before $a_{i+1}$.
\end{proof}

\begin{observation}\label{obs_transition_1_2}
$t_{c_{n}} < t_{x} < t_{d_{1}}$.
\end{observation}

\begin{proof}
The fist inequality follows from the initial lay-out of stack $s_{3,n}$, where $c_{n}$ is placed on top of $x$. The second inequality follows from the target lay-out of stack $s_{4,0}$, where $x$ is placed below $d_{1}$.
\end{proof}

\begin{observation}\label{obs_phase_2}
$t_{d_{i}} < t_{e_{i}} < t_{f_{i}}$, $\forall i \in [n]$, and $t_{f_{i}} < t_{d_{i+1}}$, $\forall i \in [n{-}1]$.
\end{observation}

\begin{proof}
The first two inequalities follow from the target lay-out of stack $s_{4,i-1}$. The last inequality follows from the fact that $d_{i+1}$ is placed at the start position of $f_{i}$, which thus has to be moved before $d_{i+1}$.
\end{proof}

\begin{observation}\label{obs_transition_2_3}
$t_{f_{n}} < t_{y} < t_{g_{1}}$.
\end{observation}

\begin{proof}
These inequalities follow from the initial lay-out of stack $s_{4,n}$ and the target lay-out of stack $s_{2,1}$, respectively.
\end{proof}

\begin{observation}\label{obs_phase_3}
$t_{g_{i}} < t_{i}^{(2)}$, $\forall i \in [n]$, and $t_{i}^{(2)} < t_{h_{i}} < t_{g_{i+1}}$, $\forall i \in [n{-}1]$.
\end{observation}

\begin{proof}
The first two inequalities follow from the target lay-out of stack $s_{2,i}$. The last inequality follows from the fact that $g_{i+1}$ takes over the position of $h_{i}$. This is only possible if $h_{i}$ moves before $g_{i+1}$.
\end{proof}

Combining these observations, we can split the unique sequence of moves into three phases.

In the first phase, the following set of moves is repeated for $i \in [n]$: move $a_{i}$ to its final location in stack $s_{3,i-1}$, move $\pi(n{-}i{+}1)$ to an intermediate stack, move $b_{i}$ on top of $a_{i}$, and move $c_{i}$ on top of $b_{i}$. This sequence follows immediate from Observation \ref{obs_phase_1}. The first phase is concluded by moving container $x$ from stack $s_{3,n}$ to stack $s_{4,0}$, which follows from Observation \ref{obs_transition_1_2}.

In the second phase, we repeat the following set of moves for $i \in [n]$: move $d_{i}$ to stack $s_{4,i-1}$, move $e_{i}$ on top of $d_{i}$, and move $f_{i}$ on top of $e_{i}$. This phase follows from Observation \ref{obs_phase_2}. The final move of the second phase is to move container $y$ from stack $s_{4,n}$ to stack $s_{2,1}$. This move follows from Observation \ref{obs_transition_2_3}.

In the third phase, repeat these set of moves for $i \in [n{-}1]$: move $g_{i}$ to its final position, move $i$ on top of $g_{i}$, and move $h_{i}$ on top of $i$. Finally, move $g_{n}$ to stack $s_{2,n}$, and place $n$ on top of it. This sequence follows from Observation \ref{obs_phase_3}.

For this sequence, we will show that the following two properties, that resemble Properties \ref{prop_order} and \ref{prop_layout}, hold.

\begin{property}
\label{prop_order_configuration}
The node containers are first moved in the order $\pi(n),\ldots,\pi(1)$. For their second move, the order is $1,\ldots,n$, where the second move of $1$ occurs after the first move of $\pi(1)$
\end{property}

\begin{property}
\label{prop_layout_configuration}
All node containers are on stacks $1,\ldots,k$ after they are moved for the first time. Each of these stacks contains at most $m$ node containers.
\end{property}

Observe that Property \ref{prop_order_configuration} follows immediately from the unique order of moves. To see that Property \ref{prop_layout_configuration} also holds is not so trivial. Since the stack height is $m$, clearly at most $m$ node containers are on each of the stacks $1,\ldots,k$. Hence, we need to show that all node containers are placed on stacks $1,\ldots,k$ after they are moved once. Thereto, take an arbitrary node container $\pi(i)$. Observe that stacks that are involved in a move between the first and second move of $\pi(i)$ can not be used as intermediate stack. This immediately excludes stacks $s_{1,1}, \ldots, s_{1,n}$, $s_{2,1}, \ldots, s_{2,n}$ and $s_{4,0}, \ldots, s_{4,n}$ for use as an intermediate stack. Consider the lay-out at the moment that $\pi(i)$ is the next container to be moved for the first time. In previous moves $c_{1}, \ldots, c_{n-i}$ were moved to their final position, so that currently stacks $s_{3,0}, \ldots, s_{3,n-i-1}$ are full. In future moves, but before $\pi(i)$ is moved for a second time, $c_{n-i+1}, \ldots, c_{n}$ need to be moved to their final position. These moves involve stacks $s_{3,n-i}, \ldots, s_{3,n}$, which implies that stacks $s_{3,0}, \ldots, s_{3,n}$ can also not be used as intermediate stacks. As $\pi(i)$ was chosen arbitrarily, this holds for all node containers. This leaves stacks $1,\ldots,k$ as only candidates for intermediate stacks, so that Property \ref{prop_layout_configuration} also holds. 

Using Properties \ref{prop_order_configuration} and \ref{prop_layout_configuration}, we state the following lemma, that resembles Lemma \ref{lem_mes}.

\begin{lemma}
\label{lem_mes_configuration}
After all node containers are moved once, the lay-out of stacks $1,\ldots,k$ constitutes a solution for \exclusionscheduling
\end{lemma}

\begin{proof}
The proof is identical to the one for Lemma \ref{lem_mes}. The stacks still need to be sorted in non-increasing order when viewed from the bottom. This time \emph{not} to obtain a target lay-out, but to make sure that the order for the second move of the node containers, namely $1,\ldots,n$ is possible.
\end{proof}

Hence, by Lemma \ref{lem_mes_configuration} we have a solution for \exclusionscheduling. Assume we have a solution for \exclusionscheduling. As solution for \configurationstacking we take the same order of moves as before. For the dummy containers, which are only moved once, the exact move is thus known, as we know the origin and destination stack from the initial and target lay-out, respectively. For the node containers, which are moved twice, we have to specify an intermediate stack. As intermediate stack for $\pi(i)$ we take stack $j$ such that $i \in P_{j}$. If all moves are valid, we obtain the target lay-out by definition.

As all node containers are moved to stacks $1,\ldots,k$, they do not ``interfere'' with the moves of the dummy containers, \ie, they are not moved on top of a dummy container. Hence, all moves involving dummy containers are valid. Hence, we only need to show that the moves involving the node containers are valid. Clearly, the first move of each node container is valid, since at most $m$ node containers are placed on each stack. The second move of the node containers, which is in the order $1,\ldots,n$, is possible if stacks $1,\ldots,k$ are sorted in non-increasing order. Therefore, we state the following lemma, that resembles Lemma \ref{lem_stacking}.

\begin{lemma}
\label{lem_stacking_configuration}
After all node containers are moved once, the node containers in stacks $1,\ldots,k$ are sorted in non-increasing order when viewed from the bottom.
\end{lemma}

\begin{proof}
The proof is identical to the one for Lemma \ref{lem_stacking}.
\end{proof}

Hence, all moves are valid and we obtain the target lay-out, so that we have a solution for \configurationstacking. Concluding, we have that for every fixed $h \geq 6$ \configurationstacking is NP-hard.
\end{proof}

\section{Tables}\label{Tables}

This section contains a more elaborate overview of the results described in Section \ref{sec:expRes}. All four tables contain the same rows. The results are split into groups for the different values for Priorities (P), Stacks (S), Height (H), and Fill (F), respectively. The final row represents the results over all instances.

Table \ref{results_quality} states for each group the number of instances, the mis-overlay index in percentage, the number of instances solved, and the integrality gap for the solved instances.

Table \ref{results_procedure} contains for each group the number of solved and killed trees, and the number of moves. These results are split over the solved and unsolved instances.

Table \ref{results_time} contains time related results. For each group, the total running time, the time spend on solving the LP relaxation, the time spend on generating sequences, and the time spend on clearing the cplex models is given.

Table \ref{results_memory} contain the memory related results. With respect to the nodes, the average number solved is given, as well as the average and maximum number in memory. With respect to the sequences, the average and maximum number generated is given, as well as the maximum number in memory and in the LP model. 

\begin{table}[h]
\centering
\begin{tabular}{lc|R{1.65cm}R{1.65cm}R{1.65cm}R{1.65cm}}
\multicolumn{2}{l|}{Group} & instances & mis-overlay & solved & int. gap \\
\hline
\multirow{3}{*}{P} & 2 & 320 & 32.3 & 320 & 1.05 \\
 & 3 & 320 & 38.4 & 320 & 1.13 \\
 & 6 & 320 & 44.1 & 305 & 1.21 \\
\hline
\multirow{4}{*}{S} & 3 & 240 & 39.4 & 239 & 1.34 \\
 & 5 & 240 & 39.5 & 236 & 1.09 \\
 & 7 & 240 & 35.4 & 236 & 1.05 \\
 & 9 & 240 & 38.8 & 234 & 1.03 \\
\hline
\multirow{2}{*}{H} & 4 & 480 & 31.4 & 480 & 1.09 \\
 & 6 & 480 & 45.1 & 465 & 1.16 \\
\hline
\multirow{2}{*}{F} & 50 & 480 & 35.5 & 480 & 1.08 \\
 & 70 & 480 & 41.1 & 465 & 1.17 \\
\hline
\multicolumn{2}{l|}{Total}& 960 & 38.3 & 945 & 1.13 \\
\end{tabular}
\caption{Results (1)}
\label{results_quality}
\end{table}

\begin{table}[h]
\centering
\begin{tabular}{lc|R{1.3cm}R{1.3cm}R{1.3cm}|R{1.3cm}R{1.3cm}R{1.3cm}}
&& \multicolumn{3}{c|}{solved} & \multicolumn{3}{c}{unsolved} \\
\multicolumn{2}{l|}{Group} & T solved & T killed & moves & T solved & T killed & moves \\
\hline
\multirow{3}{*}{P} & 2 & 1.93 & 0.67 & 0.39 & - & - & - \\
 & 3 & 2.38 & 0.78 & 0.49 & - & - & - \\
 & 6 & 2.95 & 0.96 & 0.59 & 5.07 & 2.13 & 0.83 \\
\hline
\multirow{4}{*}{S} & 3 & 3.48 & 1.00 & 0.67 & 9.00 & 2.00 & 1.58 \\
 & 5 & 2.27 & 0.72 & 0.47 & 6.25 & 2.25 & 0.93 \\
 & 7 & 1.94 & 0.67 & 0.39 & 5.00 & 2.50 & 0.78 \\
 & 9 & 1.95 & 0.81 & 0.42 & 3.67 & 1.83 & 0.68 \\
\hline
\multirow{2}{*}{H} & 4 & 1.86 & 0.50 & 0.39 & - & - & - \\
 & 6 & 2.98 & 1.12 & 0.58 & 5.07 & 2.13 & 0.83 \\
\hline
\multirow{2}{*}{F} & 50 & 1.76 & 0.42 & 0.43 & - & - & - \\
 & 70 & 3.08 & 1.19 & 0.55 & 5.07 & 2.13 & 0.83 \\
\hline
\multicolumn{2}{l|}{Total} & 2.41 & 0.80 & 0.49 & 5.07 & 2.13 & 0.83 \\
\end{tabular}
\caption{Results (2)}
\label{results_procedure}
\end{table}

\begin{table}[h]
\centering
\begin{tabular}{lc|R{1.6cm}R{1.6cm}R{1.6cm}R{1.6cm}}
\multicolumn{2}{l|}{Group} & total time & lp time & gen time & clear time \\
\hline
\multirow{3}{*}{P} & 2 & 0.72 & 0.24 & 0.31 & 0.05 \\
 & 3 & 5.04 & 2.07 & 2.41 & 0.39 \\
 & 6 & 275.18 & 137.69 & 100.80 & 34.59 \\
\hline
\multirow{4}{*}{S} & 3 & 45.03 & 18.55 & 24.14 & 2.08 \\
 & 5 & 79.26 & 44.37 & 24.43 & 10.17 \\
 & 7 & 98.88 & 52.15 & 31.36 & 14.60 \\
 & 9 & 151.41 & 71.59 & 58.10 & 19.85 \\
\hline
\multirow{2}{*}{H} & 4 & 0.59 & 0.20 & 0.27 & 0.04 \\
 & 6 & 186.70 & 93.12 & 68.74 & 23.31 \\
\hline
\multirow{2}{*}{F} & 50 & 2.38 & 0.77 & 1.24 & 0.20 \\
 & 70 & 184.92 & 92.56 & 67.78 & 23.15 \\
\hline
\multicolumn{2}{l|}{Total} & 93.65 & 46.66 & 34.51 & 11.67 \\
\end{tabular}
\caption{Results (3)}
\label{results_time}
\end{table}

\begin{table}[h]
\centering
\begin{tabular}{lc|R{1.0cm}R{0.8cm}R{0.8cm}|R{1.24cm}R{1.43cm}R{1.24cm}R{1.24cm}}
&& \multicolumn{3}{c|}{nodes} & \multicolumn{4}{c}{sequences} \\
\multicolumn{2}{l|}{Group} & solve & avg & max & avg & max & mmry & model \\
\hline
\multirow{3}{*}{P} & 2 & 35.3 & 3.4 & 36 & 481 & 38,335 & 4,979 & 2,677 \\
 & 3 & 72.7 & 4.6 & 50 & 2,580 & 148,267 & 26,911 & 16,354 \\
 & 6 & 779.1 & 6.3 & 56 & 40,128 & 1,491,758 & 59,729 & 39,303 \\
\hline
\multirow{4}{*}{S} & 3 & 172.5 & 2.1 & 14 & 22,386 & 1,491,758 & 59,729 & 39,303 \\
 & 5 & 181.6 & 3.9 & 30 & 11,218 & 592,369 & 30,759 & 23,495 \\
 & 7 & 330.5 & 5.4 & 39 & 11,107 & 474,298 & 19,301 & 14,009 \\
 & 9 & 498.3 & 7.7 & 56 & 12,875 & 591,631 & 22,237 & 16,048 \\
\hline
\multirow{2}{*}{H} & 4 & 13.6 & 2.8 & 28 & 543 & 68,279 & 20,704 & 14,195 \\
 & 6 & 577.8 & 6.8 & 56 & 28,250 & 1,491,758 & 59,729 & 39,303 \\
\hline
\multirow{2}{*}{F} & 50 & 57.2 & 3.3 & 31 & 1,152 & 66,710 & 19,536 & 13,844 \\
 & 70 & 534.2 & 6.3 & 56 & 27,641 & 1,491,758 & 59,729 & 39,303 \\
\hline
\multicolumn{2}{l|}{Total} & 295.7 & 4.8 & 56 & 14,396 & 1,491,758 & 59,729 & 39,303 \\
\end{tabular}
\caption{Results (4)}
\label{results_memory}
\end{table}

\end{document}